\documentclass[11pt,a4paper,reqno]{amsart}
\usepackage{multirow}
\usepackage[centertags]{amsmath}
\usepackage{amsfonts}
\usepackage{amssymb}
\usepackage{amsthm}
\usepackage{newlfont}
\usepackage{a4}
\usepackage{enumerate}
\usepackage[pdftex]{graphicx,color}
\theoremstyle{definition}
\newtheorem{defn}{Definition}[section]
\newtheorem{thm}[defn]{Theorem}

\newtheorem{tvr}[defn]{Proposition}
\newtheorem{cor}[defn]{Corollary}
\theoremstyle{remark}
\newtheorem{example}{Example}[section]
\newcommand{\id}{\mathfrak{1}}
\usepackage{bbm}
\usepackage{cite}
\newlength{\defbaselineskip}
\newcommand{\setlinespacing}[1]%
           {\setlength{\baselineskip}{#1 \defbaselineskip}}

\renewcommand{\i}{\mathrm{i}}
\newcommand{\map}{\rightarrow}

\newcommand{\q}{\quad}

\renewcommand{\epsilon}{\varepsilon}

\newcommand{\ep}{\varepsilon}
\newcommand{\la}{\lambda}
\newcommand{\al}{\alpha}
\newcommand{\om}{\omega}
\renewcommand{\rho}{\varrho}
\renewcommand{\phi}{\varphi}

\newcommand{\R}{{\mathbb{R}}}

\newcommand{\N}{{\mathbb N}}

\newcommand{\Com}{{\mathbb C}}
\newcommand{\Z}{\mathbb{Z}}
\newcommand{\C}{\mathbb{C}}

\newcommand{\set}[2]{\left\{#1 \, |\, #2 \right\}}
\newcommand{\setb}[2]{\left\{#1 \, \mid\, #2 \right\}}
\newcommand{\setm}[2]{\left\{#1 \,\, \big|\,\, #2 \right\}}
\newcommand{\abs}[1]{\left\vert#1\right\vert}
\newcommand{\wt}{\widetilde}

\newcommand{\sca}[2]{\langle #1,\, #2\rangle}
\newcommand{\comb}[2]{\begin{pmatrix}
     #1\\
     #2
  \end{pmatrix}}

\begin{document}

\title[Weight-Lattice Discretization]
{Weight-Lattice Discretization of Weyl-Orbit Functions}

\author[J. Hrivn\'{a}k]{Ji\v{r}\'{i} Hrivn\'{a}k$^{1}$}
\author[M. A. Walton]{Mark A. Walton$^{2}$}

\date{\today}
%%%%%%%%%%%%%%%%%%%%%%%%%%%%%%%%%%%%%%%%%%%%%%%%%%%%%%%%%%%
\begin{abstract}\small\ 
Weyl-orbit functions have been defined for each simple Lie algebra, and permit Fourier-like analysis on the fundamental region of the corresponding affine Weyl group.  They have also been discretized, using a refinement of the coweight lattice, so that digitized data on the fundamental region can be Fourier-analyzed. The discretized orbit function has arguments that are redundant if related by the affine Weyl group, while its labels, the Weyl-orbit representatives, invoke the dual affine Weyl group.  Here we discretize the orbit functions in a novel way, by using the weight lattice.  A cleaner theory results, with symmetry between the arguments and labels of the discretized orbit functions.  Orthogonality of the new discretized orbit functions is proved, and leads to the construction of unitary, symmetric matrices with Weyl-orbit-valued elements.  For one type of orbit function, the matrix coincides with the Kac-Peterson modular $S$ matrix, important for Wess-Zumino-Novikov-Witten conformal field theory.   

\end{abstract}

\maketitle
\noindent
$^1$ Department of Physics, Faculty of Nuclear Sciences and Physical Engineering, Czech Technical University in Prague, B\v{r}ehov\'a~7, CZ-115 19 Prague, Czech Republic\\
$^2$ Department of Physics and Astronomy, University of Lethbridge, Lethbridge, Alberta, T1K 3M4, Canada
\vspace{10pt}

\noindent
\textit{E-mail:} jiri.hrivnak@fjfi.cvut.cz, walton@uleth.ca

\medskip
\noindent
\textit{Keywords:} discretized Weyl-orbit functions, Wess-Zumino-Novikov-Witten conformal field theory, discrete orthogonality, affine modular data, Kac-Peterson matrices
%\tableofcontents

%%%%%%%%%%%%%%%%%%%%%%%%%%%%%%%%%%%%%%%%%%%%%%%%%%%%%%%%%%%%%%%
%\setlinespacing{0.98}
\section{Introduction}
%%%%%%%%%%%%%%%%%%%%%%%%%%%%%%%%%%%%%%%%%%%%%%%%%%%%%%%%%%%%%%

Weyl-orbit functions \cite{KP1, KPtrig, KP2, KP3} have been defined for all simple Lie algebras $X_n$ ($n$ is the rank, and $X=A, B, C, D, E, F$ or $G$).  They give rise to various Fourier-like analyses of data on the fundamental region $F$ of the corresponding affine Weyl groups. For the purposes of this introduction, we will not distinguish between the fundamental region $F$ and related regions $F^\sigma$ -- see Sect.~3 and eqn.~(\ref{FsFlex}).

Discretized versions have been studied for the analysis of digitized data on the same fundamental region. The digitization is controlled by a positive integer $M$, which can be considered the resolution.  The discretized Weyl-orbit functions \cite{HP, HP2} have as their domain a fragment $F_M$ of the lattice $P^\vee/M$ in $F$, where $P^\vee$ is the coweight lattice of $X_n$. Let us call $F_{P,M}$ a $P/M$-fragment of $F$.  For fixed $M$, a set of Weyl orbits can be specified such that the associated discretized Weyl-orbit functions are pairwise orthogonal \cite{MP2, HP, HMP}. Their number equals the number of points in the fragment $F_M$. The dominant-weight  representatives of the Weyl orbits fill out a  $P$-fragment of the $M$-dilation of  $F^\vee$, the fundamental region of the {\it dual} affine Weyl group associated with $X_n$.  

The properties of the orbit functions have been intensively investigated \cite{KP1, KPtrig, KP2, KP3, HaHrPa2, sin, HP2d, CzHr}, including their discretizations \cite{HP, HMP, MP2, HP2}.  Recently, a remarkable similarity was noticed between one class of discretized orbit functions with important objects in conformal field theory \cite{HW} (see eqn.~(27) therein).  Specifically, the modular $S$-matrix of the Wess-Zumino-Novikov-Witten (WZNW) conformal field theories, first written by Kac and Peterson \cite{KP}, involves an alternating sum over the Weyl orbit of a dominant weight, as does the so-called $S$-function.  In \cite{HP}, this similarity was exploited by using known attributes of the Kac-Peterson matrices to uncover new, analogous properties of the discretized Weyl-orbit functions. 

However, there is one striking, important difference between the discretized $S$-functions and the Kac-Peterson modular $S$-matrices (also known as the affine modular $S$ matrices).  While a Kac-Peterson matrix is symmetric, there is asymmetry between the labels and arguments of the orbit functions.  Arguments lie in the $P^\vee/M$-fragment $F_M$ of $F$, while the labels are elements of the $P$-fragment of $MF^\vee$.  The affine Weyl group of $X_n$ is relevant for the arguments, while for the labels, it is the dual affine Weyl group. 

Here we point out that a different discretization of the Weyl-orbit functions removes this argument-label asymmetry. This new discretization is somewhat more natural than the original one introduced in \cite{HP}, and the resulting values of the new $S$-functions are identical to the elements of the Kac-Peterson matrices. Only the starting point needs to be changed: we use a $P/M$-fragment $F_{P,M}$  as the discretization of $F$, instead of the $P^\vee/M$-fragment $F_M$ of \cite{HP}. The latter  should perhaps now be denoted $F_{P^\vee,M}$.  

Although part of our motivation comes from the appearance of one class of orbit functions  in conformal field theory, all orbit functions have cleaner, more symmetric properties with the new, finer weight-lattice discretization.  For completeness we treat together here all the different classes of orbit functions related to sign homomorphisms of the Weyl group of $X_n$ (see Sect.~3 below). As an interesting spin-off, we obtain a generalization of the Kac-Peterson formula for arbitrary sign homomorphism, eqn.~(\ref{sigmaKP}).  

Let us describe the plan of this paper.  In the next section, notation is set and properties of the relevant affine Weyl groups are reviewed. Section 3 describes the sign homomorphisms and the orbit functions built using them. Section 4 counts the numbers of arguments and labels of the discretized orbit functions as grid elements and the orthogonality of the $P$-discretized orbit functions is established. Section 5 treats the affine, or Kac-Peterson,  modular $S$-matrix, and establishes the identity of its elements with the values of weight-discretized Weyl-orbit functions.  The discrete transforms are also discussed.   Sect.~6 is our conclusion, which includes a short comparison between the new $P$- and previous $P^\vee$-discretizations of orbit functions.

%%%%%%%%%%%%%%%%%%%%%%%%%%%%%%%%%%%%%%%%%%%%%%%%%%%%%%%%%%%%%%%
\vskip0.5cm\section{Pertinent properties of affine Weyl groups}
%%%%%%%%%%%%%%%%%%%%%%%%%%%%%%%%%%%%%%%%%%%%%%%%%%%%%%%%%%%%%%%

%%%%%%%%%%%%%%%%%%%%%%%%%%%%%%%%%%%%%%%%%%%%%%%%%%%%%%%%%%%%%%%
\subsection{Roots and weight lattice}\

The notation, established in \cite{HP}, is used. Recall that, to the simple Lie algebra of rank $n$, corresponds the set of simple roots $\Delta=\{\al_1,\dots,\al_n\}$ of  the root system $\Pi$ \cite{BB,Bour,H2}. The set $\Delta$ spans the Euclidean space $\R^n$, with the scalar product denoted by $\sca{\,}{\,}$. The set of simple roots determines partial ordering $\leq$ on $\R^n$ -- for $\la,\nu \in \R^n$ it holds that  $\nu\leq \la$ if and only if $\la-\nu = k_1\al_1+\dots+ k_n \al_n$ with $k_i \in \Z_{\geq 0}$ for all $i\in I$ with $I:=\{1,\dots,n\}$.   The root system $\Pi$ and its set of simple roots $\Delta$ can be defined independently of Lie theory and such sets which correspond to compact simple Lie groups are called crystallographic \cite{H2}. There are two types of sets of simple roots -- the first type with roots of only one length (simply-laced), denoted conventionally as $A_{n\geq 1}$,  $D_{n\geq 4}$, $E_6$, $E_7$, $E_8$; and the second type
with two different lengths of roots, denoted $B_{n\geq 3}$, $C_{n\geq 2}$, $G_2$ and $F_4$. For the second type systems, the set of simple roots consists of short simple roots $\Delta_s$ and long simple roots $\Delta_l$, i.e. the following disjoint decomposition is given,
\begin{equation}\label{sl}
\Delta=\Delta_s\cup\Delta_l.
\end{equation}
The standard objects, related to the set $\Delta\subset \Pi$, are the following \cite{BB,H2,DMS}:
\begin{itemize}
\item the highest root $\xi \in \Pi$ with respect to the partial ordering $\leq$ restricted on $\Pi$,
\item the marks $m_1,\dots,m_n\in \N$ of the highest root $\xi =: -\al_0=m_1\al_1+\dots+m_n\al_n$, together with $m_0:=1$, the marks are summarized in Table 1 in \cite{HP},
\item
the root lattice $Q=\Z\al_1+\dots+\Z\al_n $,
\item
the $\Z$-dual lattice to $Q$, the coweight lattice 
\begin{equation*}
 P^{\vee}=\set{\om^{\vee}\in \R^n}{\sca{\om^{\vee}}{\al}\in\Z,\, \forall \al \in \Delta}=\Z \om_1^{\vee}+\dots +\Z \om_n^{\vee},
 \end{equation*}
 where $\{\omega_1^\vee,\ldots,\omega_n^\vee\}$ are the fundamental coweights, 
with 
\begin{equation}\label{aldom} 
\sca{\al_i}{ \om_j^{\vee}}=\delta_{ij},
\end{equation}
\item
the dual root lattice $Q^{\vee}=\Z \al_1^{\vee}+\dots +\Z \al^{\vee}_n$, where $\al^{\vee}_i=2\al_i/\sca{\al_i}{\al_i}$, $i\in I$,
\item
the comarks $q_1, \dots ,q_n\in \N$ of the highest root $\xi= q_1\al_1^{\vee} + \dots + q_n \al_n^{\vee}$ and $q_0:=1$; for simply-laced root systems the marks and the comarks coincide, for nonsimply-laced systems the comarks are summarized in Table \ref{tab},
\item
the Coxeter number $m=\sum_{i\in \hat I} m_i$ and the dual Coxeter number $g=\sum_{i\in \hat I} q_i$, where $ \hat I := \{0,1,\dots, n \}$,
\item the $\Z$-dual lattice to $Q^\vee$, the weight lattice
\begin{equation*}
 P=\set{\om\in \R^n}{\sca{\om}{\al^{\vee}}\in\Z,\, \forall \al^{\vee} \in Q^\vee}=\Z \om_1+\dots +\Z \om_n,
\end{equation*} 
where $\{\omega_1,\ldots,\omega_n\}$ are the fundamental weights, with 
\begin{equation}\label{alvdom} 
\sca{\al_i^{\vee}}{ \om_j}=\delta_{ij},
\end{equation}
\item
the Cartan matrix $C$ with elements $C_{ij}=\sca{\al_i}{\al^{\vee}_j}$
and with the properties
\begin{equation}\label{Cartanm}
 \al^\vee_j = \sum_{k\in I}  \om_k^{\vee} C_{kj}\ ,\quad  \alpha_j = \sum_{k\in I}  C_{jk}  \om_k \  ,
\end{equation}
\item the Gram determinant of the $\al^\vee$-basis $d=\det \sca{\al^{\vee}_i}{\al^{\vee}_j}$ determining the orders of the quotient group $P/Q^\vee $ 
 \begin{equation}\label{dnum}
 d=|P/Q^\vee |=\frac{2}{\langle \alpha_1,\alpha_1\rangle}\cdots\frac{2}{\langle \alpha_n,\alpha_n\rangle} \det C ;
\end{equation}
for simply-laced systems $d$ concides with $\det C$, for nonsimply-laced systems the numbers $d$ are summarized in Table \ref{tab}.
\end{itemize}
{
\begin{table}
\begin{tabular}{|c||c|c|c|c|c|c|}
\hline
Type           & $d$ & Comarks $q_1,\dots,q_n$  & $R^{\sigma^s}$  & $R^{\sigma^l}$ &  $q^{\sigma^s}$ & $q^{\sigma^l}$ \\
\hline\hline
$B_n\ (n\geq3)$& $4$ & $1,2,\dots,2,1$  & $r_n$ & $r_0,\, r_1, \dots, r_{n-1}$ & $1$ & $2n-2$ \\ \hline
$C_n\ (n\geq2)$& $2^n$ & $1,1,\dots,1,1$  & $r_1, \dots, r_{n-1}$ & $r_0,\,r_n$  & $n-1$ & $2$ \\ \hline
$G_2$ &      $3$     &    $2,1$    &$r_2$     & $r_0,\,r_1$ &   $1$ & $3$ \\ \hline
$F_4$ &    $4$      &    $2,3,2,1$              &$r_3,r_4$ & $r_0,\,r_1,\, r_2$ & $3$ & $6$\\ \hline
\end{tabular}
\medskip
\caption{The orders of $|P/Q^\vee |$, the comarks, the decomposition of the sets of generators $R$ and of the dual Coxeter number $g$ of nonsimply-laced root systems. Numbering of the simple roots is standard (see e.g. Figure 1 in \cite{HP}).}\label{tab}
\end{table}
}

The $n$ reflections $r_\al$, $\al\in\Delta$ in $(n-1)$-dimensional mirrors orthogonal to simple roots intersecting at the origin are given explicitly for $a\in\R^n$ by
\begin{equation}\label{refl}
r_{\al}a=a-\sca{\al}{a}\al^\vee. 
\end{equation}
and the affine reflection $r_0$ with respect to the highest root $\xi$ is given by
\begin{equation}\label{aWeyl}
r_0 a=r_\xi a + \frac{2\xi}{\sca{\xi}{\xi}}\,,\qquad
r_{\xi}a=a-\frac{2\sca{a}{\xi} }{\sca{\xi}{\xi}}\xi\,,\qquad a\in\R^n\,.
\end{equation}
%For any positive integer $M$ the affine reflection $ r_{0,M}$ is defined by
%\begin{equation}
%r_{0,M}\, a = r_{\xi} a + M\frac{2\xi}{\sca{\xi}{\xi}}.
%\end{equation}
The set of reflections $r_1 := r_{\al_1}, \, \dots, r_n := r_{\al_n}$, together with the affine reflection $r_0$,  is denoted by $R$,
\begin{equation}\label{R}
R=\{ r_0,r_1,\dots,r_n \}.
\end{equation}
%and the set of reflections $r_1, \dots, r_n$, together with the affine reflection $ r_{0,M}$,  is denoted by $R_M$,
%\begin{equation}\label{R}
%R_M=\{ r_{0,M},r_1,\dots,r_n \}.
%\end{equation}

\subsection{Weyl group and affine Weyl group}\

The Weyl group $W$ is generated by $n$ reflections $r_\al$, $\al\in\Delta$. The set $R$ of $n+1$ generators \eqref{R} generates the affine Weyl group $W^{\mathrm{aff}}$. Since the affine Weyl group $W^{\mathrm{aff}}$ is the semidirect product of the Abelian group of translations $T(Q^\vee)$ by shifts from $Q^\vee$ and of the Weyl group~$W$,
\begin{equation*}%\label{direct}
 W^{\mathrm{aff}}= T(Q^\vee) \rtimes W = \langle r \mid r\in R \rangle\, ,
\end{equation*}
for any $w^{\mathrm{aff}}\in W^{\mathrm{aff}}$, there exists a unique $w\in W$ and a unique shift $T(q^{\vee})$ such that $w^{\mathrm{aff}}=T(q^{\vee})w$.
Taking any $w^{\mathrm{aff}}=T(q^\vee)w\in {W}^{\mathrm{aff}}$, the retraction homomorphism $\psi:{W}^{\mathrm{aff}}\map W$  %and the mapping $\tau:{W}^{\mathrm{aff}}\map Q^\vee$ are given by
is given by
\begin{align}
\psi(w^{\mathrm{aff}}) &=w, \label{ret}%\\
%\tau(w^{\mathrm{aff}}) &=q^\vee \label{retq}.
\end{align}

The fundamental domain $F$ of $W^{\mathrm{aff}}$, which consists of precisely one point of each $W^{\mathrm{aff}}$-orbit, is the convex hull of the points $\left\{ 0, \frac{\om^{\vee}_1}{m_1},\dots,\frac{\om^{\vee}_n}{m_n} \right\}$,
\begin{align}
F &=\setm{\sum_{i\in I} y_i\om^{\vee}_i}{y_j\in \R_{\geq 0},\, \sum_{i\in \hat I}y_j m_j=1  }. \label{deffun}
\end{align}
Let us denote the isotropy subgroup of a point $a\in\R^n$ and its order by
\begin{equation*}
\mathrm{Stab}_{W^{\mathrm{aff}}}(a) = \setb{w^{\mathrm{aff}}\in W^{\mathrm{aff}}}{w^{\mathrm{aff}}a=a},\q h(a)=|\mathrm{Stab}_{W^{\mathrm{aff}}}(a)|,
\end{equation*}
and define a function $\ep:\R^n\map\N$ by the relation
\begin{equation}\label{epR}
\ep(a)=\frac{|W|}{h(a)}.
\end{equation}
The following abbreviation for any $M\in \N$ is used 
\begin{equation}\label{hM}
h_M(a)=\abs{\mathrm{Stab}_{W^{\mathrm{aff}}}\left(\frac{a}{M}\right)}.
\end{equation}
 
Since for any $w^{\mathrm{aff}}\in W^{\mathrm{aff}}$ the stabilizers $\mathrm{Stab}_{W^{\mathrm{aff}}}(a) $ and $\mathrm{Stab}_{W^{\mathrm{aff}}}(w^{\mathrm{aff}}a) $ are  conjugate, one obtains that
\begin{equation}\label{epshift}
\ep(a)=\ep(w^{\mathrm{aff}}a),\q w^{\mathrm{aff}}\in W^{\mathrm{aff}}.
\end{equation}

Recall that the stabilizer $\mathrm{Stab}_{W^{\mathrm{aff}}}(a)$
of a point $a=y_1\om^{\vee}_1+\dots+y_n\om^{\vee}_n\in F$ is trivial, $\mathrm{Stab}_{W^{\mathrm{aff}}}(a)=1$ if the point $a$ is in the interior of $F$, $a\in \mathrm{int}(F)$. Otherwise the group $\mathrm{Stab}_{W^{\mathrm{aff}}}(a)$
is generated by such $r_i$ for which $y_i=0$, $i\in \hat I$.

Considering the standard action of $W$ on the torus $\R^n/Q^{\vee}$, we denote for $x\in \R^n/Q^{\vee}$ the isotropy group by $\mathrm{Stab} (x)$
and the orbit and its order by
\begin{equation*}
W x=\set{wx\in \R^n/Q^{\vee} }{w\in W},\q \wt \ep(x)\equiv |Wx|.
\end{equation*}
Recall the following three properties from Proposition 2.2 in \cite{HP} of the action of $W$ on the torus $\R^n/Q^{\vee}$:
\begin{enumerate}
\item For any $x\in \R^n/Q^{\vee}$, there exists $x'\in F \cap \R^n/Q^{\vee} $ and $w\in W$ such that
\begin{equation}\label{rfun1}
 x=wx'.
\end{equation}
\item If $x,x'\in F \cap \R^n/Q^{\vee} $ and $x'=wx$, $w\in W$, then
\begin{equation}\label{rfun2}
 x'=x=wx.
\end{equation}
\item If $x\in F \cap \R^n/Q^{\vee} $, i.e. $x=a+Q^{\vee}$, $a\in F$, then $\psi (\mathrm{Stab}_{W^{\mathrm{aff}}}(a))=\mathrm{Stab}(x)$ and
\begin{equation}\label{rfunstab}
\mathrm{Stab} (x) \cong \mathrm{Stab}_{W^{\mathrm{aff}}}(a).
\end{equation}
\end{enumerate}
From \eqref{rfunstab} we obtain that for $x=a+Q^{\vee}$, $a\in F$ it holds that
\begin{equation}\label{ept}
	\ep(a)= \wt\ep(x).
\end{equation}
Note that instead of $\wt\ep(x)$, the symbol $\ep(x)$ is used for $|Wx|$, $x\in F\cap\R^n/Q^{\vee} $ in \cite{HP,HMP}. The method of calculation of the coefficients $\ep(x)$ is detailed in \S 3.7 in \cite{HP}.

\vskip0.5cm\section{Sign homomorphisms and orbit functions}

\subsection{Sign homomorphisms}\

To introduce various classes of orbit functions, we consider `sign' homomorphisms $\sigma:W\rightarrow\{\pm1 \}.$ The following two choices of homomorphism values of generators $r_\al,\,\al\in\Delta$,  lead to the well-known homomorphisms:
\begin{align}
\id(r_\al)&=1 \label{ghomid} \\
\sigma^e(r_\al)&=-1 \label{ghome}
\end{align}
which yield for any $w\in W$
\begin{align}
\id(w)&=1  \\
\sigma^e(w)&=\det w. \label{parity}
\end{align}

Using the decomposition \eqref{sl}, two additional homomorphisms are given as follows \cite{HMP}:
\begin{align}
\sigma^s(r_\al)&=\begin{cases} 1,\quad \al\in \Delta_l  \\ -1,\quad \al\in \Delta_s\end{cases}\\
\sigma^l(r_\al)&=\begin{cases} 1,\quad \al\in \Delta_s  \\ -1,\quad \al\in \Delta_l.\end{cases}
\end{align}

\subsection{Fundamental domains}\

Each of the four sign homomorphisms determines a decomposition of the fundamental domain~$F$. The factors of this decomposition are crucial for the study of the discretized orbit functions. For each sign homomorphism $\sigma$ the appropriate subset $F^\sigma\subset F$ is  
\begin{align}\label{Fsigma}
F^\sigma&=\setm{a\in F}{\sigma \circ \psi \left(\mathrm{Stab}_{W^{\mathrm{aff}}}(a)\right)=\{1\} }
\end{align}
where $\psi$ is the retraction homomorphism \eqref{ret}. Since for all points of the interior of $F$ the stabilizer is trivial, i.e. $\mathrm{Stab}_{W^{\mathrm{aff}}}(a)=1$, $a\in \mathrm{int} (F)$, the interior $\mathrm{int} (F)$ is a subset of all $F^\sigma$. 
Let us also define the corresponding subset $R^\sigma$ of generators $R$ of ${W}^{\mathrm{aff}}$
\begin{align}\label{Rsigma}
R^\sigma&=\setb{r\in R}{ \sigma \circ \psi \left( r\right)=-1 }
\end{align}
and subsets of the boundaries $H^\sigma$ of $F$
\begin{align}\label{HRsigma}
H^\sigma&=\set{a\in F}{(\exists r\in R^\sigma)(ra=a)}.
\end{align}
Note  that the sets $F^{\sigma^s}$, $R^{\sigma^s}$ and $H^{\sigma^s}$ correspond to the sets $F^s$, $R^s$ and $H^s$ from \cite{HMP}; similar correspondence holds for the long versions of these sets.
Note also that $R^\id =\emptyset$ and $R^{\sigma^e} =R$ and for nonsimply-laced systems $R= R^{\sigma^s}\cup R^{\sigma^l}$. For nonsimply-laced systems are the subsets of generators $R^{\sigma^s}$ and $R^{\sigma^l}$ summarized in Table \ref{tab}. 
 Moreover, the sets $F^{\sigma}$ and $H^{\sigma}$ correspond to the sets $F^{\sigma}(0)$ and $H^{\sigma}(0)$ from \cite{CzHr}. Thus, specializing Proposition 2.7 from \cite{CzHr}, the following set equality holds
\begin{equation}
F^\sigma=F\setminus H^\sigma.
\end{equation} 
Introducing the symbols $y^\sigma_i$, $i\in \hat I$
\begin{equation*}
y^{\sigma}_i \in \begin{cases}\R_{> 0},\q   r_i\in R^\sigma\\ \R_{\geq 0},\q r_i\in R\setminus R^\sigma,\end{cases}
\end{equation*}
the explicit form of $F^\sigma$ is given by
\begin{equation}\label{FsFlex}
F^\sigma=\setm{\sum_{i\in I}y^{\sigma}_i\om^{\vee}_i}{\sum_{i\in \hat I} y^{\sigma}_im_i=1  }.
\end{equation}

\subsection{Orbit functions}\

Depending on the type of  root system, two or four sign homomorphisms induce the corresponding types of families of complex orbit functions. Within each family are the complex functions $\phi^\sigma_\la:\R^n\map \C$ labeled by weights $\la\in P$,
\begin{equation}\label{genorb}
\phi^\sigma_\la(a)=\sum_{w\in W}\sigma (w)\, e^{2 \pi i \sca{ w\la}{a}},\q a\in \R^n.
\end{equation}
Recall from Proposition 3.1 in \cite{CzHr} that for any $w^{\mathrm{aff}} \in {W}^{\mathrm{aff}}$ and $a\in \R^n$ it holds that  
\begin{equation}\label{Sssym}
\phi^\sigma_{\la}(w^{\mathrm{aff}}a)=\sigma \circ \psi(w^{\mathrm{aff}})\phi^\sigma_{\la}(a)
\end{equation}
and that the functions $\phi^\sigma_{\la}$ are all zero on the boundary $H^\sigma$,
\begin{equation}\label{Fss}
\phi^\sigma_{\la}(a')=0,\q  a'\in H^\sigma
\end{equation}
and therefore the functions $\phi^\sigma_{\la}$ are considered on the domain $F^\sigma$ only.

\begin{tvr}\label{disklab} 
Let $a\in \frac{1}{M}P$ with $M\in \N$. Then for any $w^{\mathrm{aff}} \in {W}^{\mathrm{aff}}$ and $\la\in P$ it holds that 
\begin{equation}\label{dSssym}
\phi^\sigma_{Mw^{\mathrm{aff}}(\frac{\la}{M})}(a)=\sigma \circ \psi(w^{\mathrm{aff}})\phi^\sigma_{\la}(a)
\end{equation}
and the functions $\phi^\sigma_{\la}$ are identically zero on the boundary $MH^\sigma$, i.e.
\begin{equation}\label{dFss}
\phi^\sigma_{\la}\equiv 0,\q  \la\in MH^{\sigma}.
\end{equation}
\end{tvr}
\begin{proof}
Considering an element of the affine Weyl group of the form $w^{\mathrm{aff}}a=w'a+q^\vee$, with $q^\vee\in Q^\vee$, $a=\mu/M$ and $\mu\in P$, the property \eqref{dSssym} is derived, 
\begin{align*}
\phi^\sigma_{Mw^{\mathrm{aff}}(\frac{\la}{M})}(a)&=\sum_{w\in W}\sigma (w)\, e^{2 \pi i \sca{ w'\la+M q^\vee}{wa}}=\sum_{w\in W}\sigma (w)\, e^{2 \pi i \sca{ w'\la}{wa}}  e^{2 \pi i \sca{  q^\vee}{w\mu}}\\
   &= \sigma(w')\phi^\sigma_{\la}(a),
\end{align*}
where the third equality follows from $W$-invariance of $P$ and the $\Z$-duality of $P$ and $Q^\vee$ which implies $\sca{  q^\vee}{w\mu}\in \Z$. Specialization of property \eqref{dSssym} for the generators $r\in R^{\sigma}$ in \eqref{HRsigma}  implies for the weights $\la\in MH^{\sigma}$ that 
\begin{align*}
\phi^\sigma_{\la}(a)&=\phi^\sigma_{Mr(\frac{\la}{M})}(a)=-\phi^\sigma_{\la}(a).
\end{align*}
\end{proof}

Suppose we have $M\in \N$ and $a\in \frac{1}{M}P$. It follows from \eqref{Sssym} and \eqref{Fss} that the discretized functions $\phi^\sigma_b$ can be considered only on the set
\begin{equation}\label{Fs}
F^\sigma_{P,M}:=\frac{1}{M}P \cap F^\sigma.
\end{equation}
Defining the set
\begin{equation}\label{Ls}
\Lambda^\sigma_{P,M}:= P \cap MF^{\sigma},
\end{equation}
Proposition \ref{disklab} implies that the functions $\phi^\sigma_\la$ on the finite set $F^\sigma_{P,M}$ can be parameterized by $\la\in\Lambda^\sigma_{P,M}$ only.

%%%%%%%%%%%%%%%%%%%%%%%%%%%%%%%%%%%%%%%%%%%%%%%%%%%%%%%%%%%%%%%%%%
\vskip0.5cm\section{Discretization of orbit functions}
%%%%%%%%%%%%%%%%%%%%%%%%%%%%%%%%%%%%%%%%%%%%%%%%%%%%%%%%%%%%%%%%%%

%%%%%%%%%%%%%%%%%%%%%%%%%%%%%%%%%%%%%%%%%%%%%%%%%%%%%%%%%%%%%%%%%%
\subsection{Number of elements of $F^\sigma_{P,M}$}\

In order to derive an explicit form of the sets $F^\sigma_{P,M}$, the points of $F^\sigma$ from \eqref{FsFlex} are rewritten in the $\om$-basis (the basis of fundamental weights $\{ \omega_i, i\in I \}$) via the relation $\om^\vee_i=2\om_i/\sca{\al_i}{\al_i}$ and substitution $u^\sigma_i = 2My^\sigma_i/\sca{\al_i}{\al_i}, \, i\in \hat I$ is used. Then taking into account the relation between the marks and the comarks $q_i = m_i \sca{\al_i}{\al_i} /2 ,\, i\in I $, one obtains 
\begin{equation}\label{FP}
F^\sigma_{P,M}=\setm{\sum_{i\in I}u^{\sigma}_i\frac{\om_i}{M}}{\sum_{i\in \hat I} u^{\sigma}_iq_i=M  },
\end{equation}
with
\begin{equation*}
u^{\sigma}_i \in \begin{cases}\N,\q   r_i\in R^\sigma\\ \Z_{\geq 0},\q r_i\in R\setminus R^\sigma.\end{cases}
\end{equation*} 
From definitions \eqref{Fs} and \eqref{Ls} follows that it holds $M F^\sigma_{P,M}= \Lambda^\sigma_{P,M}$ and thus
\begin{equation}\label{LP}
\Lambda^\sigma_{P,M}=\setm{\sum_{i\in I}u^{\sigma}_i\om_i}{\sum_{i\in \hat I} u^{\sigma}_iq_i=M  },
\end{equation} 
and 
\begin{equation}\label{complete}
|F^\sigma_{P,M}|=|\Lambda^\sigma_{P,M}|.
\end{equation}

Note that the number of points of the sets 
$$F^\sigma_{P^\vee,M}= \frac{1}{M}P^\vee \cap F^\sigma$$ 
is calculated for all cases in \cite{HP,HMP}. Since for simply-laced root systems with roots of only one length, $A_{n\geq 1}$, $D_{n\geq 4}$, $E_6$, $E_7$, $E_8$ it holds that $P=P^\vee$ and thus $F^\sigma_{P^\vee,M}=F^\sigma_{P,M}$, the formulas for $|F^\sigma_{P^\vee,M}|$, $\sigma = \id, \sigma^e$ in \cite{HP} determine also the numbers $|F^\sigma_{P,M}|$. The formulas for the non-simply laced systems are derived in the following theorem.  
\begin{thm}\label{numAn}
The numbers of points of grids $F^\id_{P,M}$ of Lie algebras $B_n$, $C_n$, $G_2$ and $F_4$ are given by the following relations.
\begin{enumerate}\item  $B_n,\,n\geq 3$,
\begin{eqnarray*}
|F^\id_{P,2k}(B_n)|=& \begin{pmatrix}n+k \\ n \end{pmatrix}+3\begin{pmatrix}n+k-1 \\ n \end{pmatrix}\\
|F^\id_{P,2k+1}(B_n)|=& 3 \begin{pmatrix}n+k \\ n \end{pmatrix}+\begin{pmatrix}n+k-1 \\ n \end{pmatrix}
\end{eqnarray*}
\item $C_n,\,n\geq 2$,
$$|F^\id_{P,M}(C_n)|=\begin{pmatrix}n+M\\ n \end{pmatrix}$$
\item $G_2$,
\begin{equation*}
\begin{alignedat}{4}
|F^\id_{P,2k}(G_2)|&= k^2+2k+1, &\qquad |F^\id_{P,2k+1}(G_2)|&= k^2+3k+2
\end{alignedat}
\end{equation*}
\item $F_4$,
\begin{equation*}
\begin{alignedat}{4}
|F^\id_{P,6k}(F_4)|&= \frac{9}{2}k^4+\frac{27}{2}k^3+14k^2+6k+1, &\qquad |F^\id_{P,6k+1}(F_4)|&= \frac92k^4+\frac{33}{2}k^3+\frac{43}{2}k^2+\frac{23}{2}k+2\\
|F^\id_{P,6k+2}(F_4)|&= \frac{9}{2} k^4+\frac{39}{2}k^3+\frac{61}{2}k^2+\frac{41}{2}k+5, &\qquad |F^\id_{P,6k+3}(F_4)|&=\frac{9}{2}k^4+\frac{45}{2}k^3+41k^2+32k+9\\
|F^\id_{P,6k+4}(F_4)|&=\frac{9}{2}k^4+\frac{51}{2}k^3+53k^2+48k+16,  &\qquad |F^\id_{P,6k+5}(F_4)|&=\frac{9}{2}k^4+\frac{57}{2}k^3+\frac{133}{2}k^2+\frac{135}{2}k+25.
\end{alignedat}
\end{equation*}
\end{enumerate}
\end{thm}
\begin{proof}
The algorithm for calculation of the counting polynomials of solutions of the equation $u^\id_0+q_1u^\id_1+\dots+ q_nu^\id_n= M$ with $u^\id_0,u^\id_1,\dots,u^\id_n\in\Z_{\geq 0} $ and $q_1,\dots, q_n\in \N$ is formulated in Proposition 3.2 in \cite{HP}. This algorithm describes for given $q_1,\dots,q_n$ the construction of $L\times N$ matrices $R^P$, with $L=\mathrm{lcm} (q_1,\dots,q_n)$ and $(n+1)L-( q_0+q_1+\dots +q_n)=LN+N',\, N,N'\in \Z^{\geq 0},\, N'<L$,   from which the counting polynomials are constructed using the formula 
\begin{equation*}
 |F^\id_{P,Lk+l}|=\sum_{i=0}^{N}R^P_{li}\comb{n-i+k}{n}.
\end{equation*}
Using the algorithm, the following $R^P$ matrices are obtained,
$R^P(B_n)=\left(\begin{smallmatrix}1&3\\\noalign{\medskip}3&1\end{smallmatrix}\right) $, $R^P(C_n)=(1)$, $R^P(G_2)=\left(\begin{smallmatrix}1&1\\\noalign{\medskip}2&0\end{smallmatrix}\right) $ and  $R^P(F_4)=\left(\begin{smallmatrix}1&34&64&9\\\noalign{\medskip}2&46&55&5\\ \noalign{\medskip}5&55&46&2\\ \noalign{\medskip}9&64&34&1\\ \noalign{\medskip}16&67&25&0\\ \noalign{\medskip}25&67&16&0\end{smallmatrix}\right) $.

\end{proof}

Each subset \eqref{Rsigma}  of the set of generators $R^\sigma \subset R$ determines also the corresponding decomposition of the sum of comarks
\begin{equation}\label{qsigma}
q^\sigma = \sum_{r_i \in R^\sigma }q_i  
\end{equation}  
Note that $q^{\sigma^e}=g $, $q^{\id}=0 $ and for nonsimply-laced systems $g=q^{\sigma^s}+q^{\sigma^l} $. The numbers $q^{\sigma^s}$ and $q^{\sigma^l} $ are for nonsimply-laced systems  tabulated in Table \ref{tab}. The number $q^\sigma$ determines the numbers of elements of $F^{\sigma^e}_{P,M}$, $F^{\sigma^s}_{P,M}$ and $F^{\sigma^l}_{P,M}$ from the counting formulas for $F^\id_{P,M}$.
\begin{tvr}\label{Cox}
For any sign homomorphism $\sigma$ and any $M\in\N$ it holds that 
\begin{equation}\label{numFsFl}
| F^{\sigma}_{P,M}|=\begin{cases}0 & M<q^\sigma \\ 1 & M=q^\sigma \\ |F^\id_{P,M-q^\sigma}|, & M>q^\sigma \end{cases}
\end{equation}
\end{tvr}
\begin{proof}
Taking non-negative numbers $u^\id_i\in \Z^{\geq 0}$ and substituting the relations $ u^\sigma_i=1+u^\id_i$ if $r_i\in R^\sigma$ and $ u^\sigma_i=u^\id_i$ if $r_i\in R\setminus R^\sigma $ into the defining relation (\ref{FP}), one gets 
\begin{equation*}
u^\id_0+q_1u^\id_1+\dots + q_n u^\id_n= M-q^\sigma,\q u^\id_0,\dots,u^\id_n\in \Z^{\geq 0}.
\end{equation*}
This equation has one solution $[0,\dots,0]$ if $M=q^\sigma$, no solution if $M<q^\sigma$, and is equal to the defining relation (\ref{FP}) of $F^\id_{M-q^\sigma}$ if $M>q^\sigma$. 
\end{proof}
Note that Proposition \ref{Cox} implies that $$| F^{\sigma}_{P,M+q^\sigma}|=|F^\id_{P,M}|. $$

%%%%%%%%%%%%%%%%%%%%%%%%%%%%%%%%%%%%%%%%%%%%%%%%%%%%%%%%%%%%%%%%%%

\begin{example}\label{ex1}
For the Lie algebra $C_2$, it holds that $d=4$, $ q^{\sigma^e}=3$, $ q^{\sigma^s}=1$ and $ q^{\sigma^l}=2$. For $M=3$, the order of the group $\frac{1}{3}P/Q^{\vee}$ is equal to $36$, and according to Theorem \ref{numAn}, Proposition \ref{Cox} and \eqref{complete} one obtains 
\begin{equation*}
\begin{alignedat}{4}
|F^\id_{P,3}(C_2)|&= |\Lambda^\id_{P,3}(C_2)|=\comb{5}{2}=10, &\qquad |F^{\sigma^e}_{P,3}(C_2)|&=|\Lambda^{\sigma^e}_{P,3}(C_2)|= 1,\\ 
|F^{\sigma^s}_{P,3}(C_2)|&=|\Lambda^{\sigma^s}_{P,3}(C_2)|= \comb{4}{2}=6, &\qquad |F^{\sigma^l}_{P,3}(C_2)|&=|\Lambda^{\sigma^l}_{P,3}(C_2)|= \comb{3}{2}=3. 
\end{alignedat}
\end{equation*}
The coset representatives of $\frac{1}{3}P/Q^{\vee}$, the fundamental domains $F^{\sigma}$ and the grids $F^\sigma_{P,3}$ are depicted in Figure~\ref{figC2}.
\begin{figure}
\includegraphics{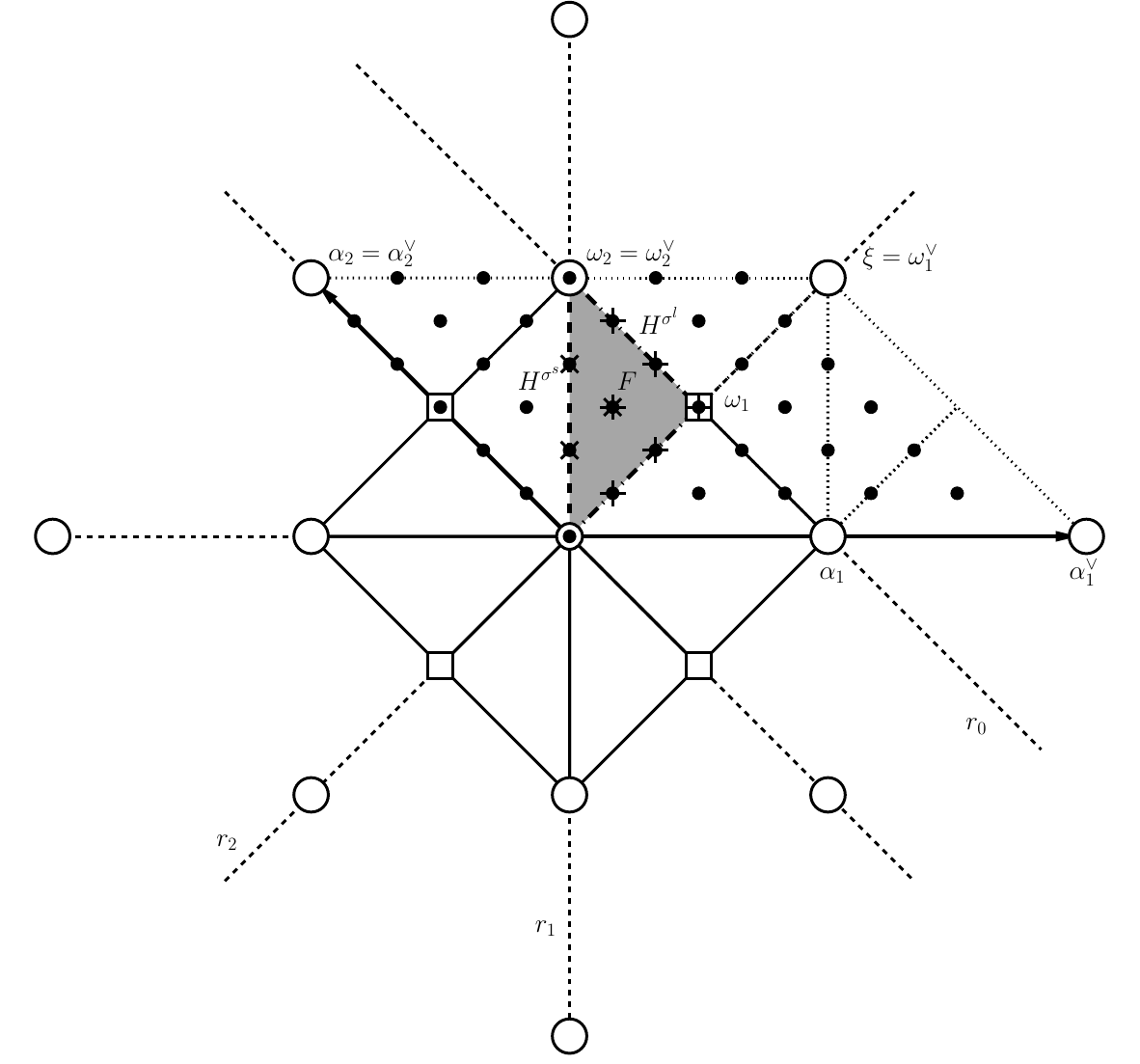}
\caption{ \textit{The fundamental domains $F^\sigma$ and grids $F^\sigma_{P,3}$ of $C_2$.} The fundamental domain $F^\id=F$ is depicted as the grey triangle containing both borders $H^{\sigma^s}$ and $H^{\sigma^l}$, depicted as the thick dashed line and dot-and-dashed lines, respectively. The coset representatives of $\frac{1}{3}P/Q^{\vee}$ are shown as $36$ black dots. The six dots belonging to $F^{\sigma^s}_{P,3}$ and three dots belonging to $F^{\sigma^l}_{P,3}$ are crossed with '$+$' and '$\times$', respectively. The dot crossed with both '$+$' and '$\times$' represents the only point of $F^{\sigma^e}_{P,3}$. The dashed lines represent 'mirrors' $r_0,r_1$ and $r_2$. Circles are elements of the root lattice $Q$; together with the squares they are elements of the weight lattice $P$.}\label{figC2}
\end{figure}
The representatives of coset $P/3Q^\vee$ together with the grids of weights $\Lambda^\sigma_{P,3}$ are depicted in Figure~\ref{figC2d}.
\begin{figure}
\includegraphics{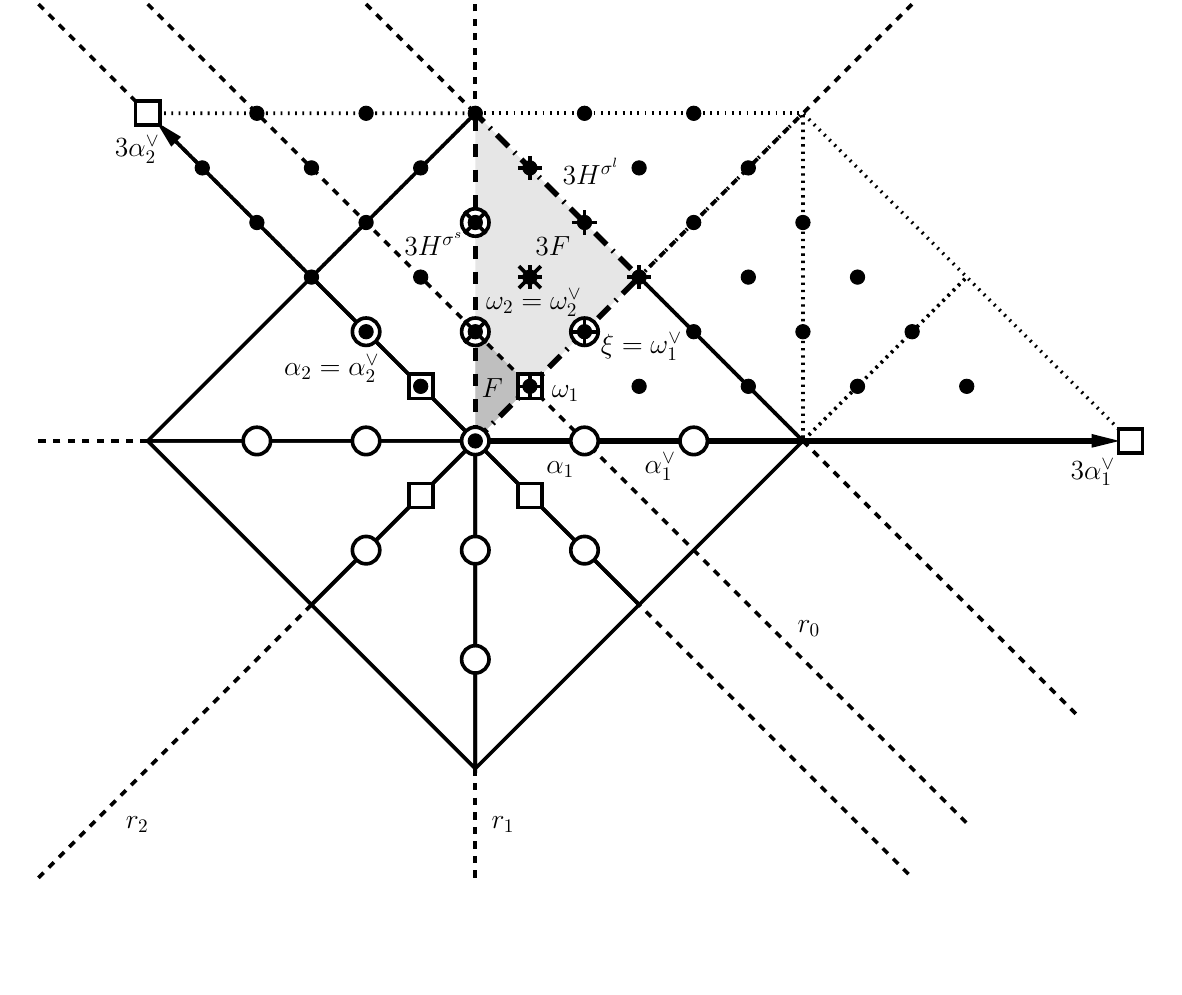}
\caption{ The grids of weights $\Lambda^\sigma_{P,3}$ of $C_2$. The darker grey triangle is the fundamental domain $F$ and the lighter grey triangle is the magnified domain $3F$. The magnified borders $3H^{\sigma^s}$ and $3H^{\sigma^l}$ are depicted as the thick dashed lines and dot-and-dashed lines, respectively. The representatives of coset $P/3Q^\vee$ of $C_2$ are shown as $36$ black dots. The six dots belonging to $\Lambda^{\sigma^s}_{P,3}$ and three dots belonging to $\Lambda^{\sigma^l}_{P,3}$ are crossed with '$+$' and '$\times$', respectively. The dot crossed with both '$+$' and '$\times$' represents the only point of $\Lambda^{\sigma^e}_{P,3}$. The circles, squares and mirrors coincide with those in Figure \ref{figC2}. }\label{figC2d}
\end{figure}
\end{example}
%%%%%%%%%%%%%%%%%%%%%%%%%%%%%%%%%%%%%%%%%%%%%%%%%%%%%%%%%%%%%%%%%%

%%%%%%%%%%%%%%%%%%%%%%%%%%%%%%%%%%%%%%%%%%%%%%%%%%%%%%%%%%%%%%%%%%
%\vskip0.5cm\section{Discrete orthogonality and transforms of orbit functions}
%%%%%%%%%%%%%%%%%%%%%%%%%%%%%%%%%%%%%%%%%%%%%%%%%%%%%%%%%%%%%%%%%%

%%%%%%%%%%%%%%%%%%%%%%%%%%%%%%%%%%%%%%%%%%%%%%%%%%%%%%%%%%%%%%%%%%
\subsection{Discrete orthogonality of orbit functions}\ 

To describe the discrete orthogonality of all four types of orbit functions on the sets $F^\sigma_{P,M}$, the ideas developed in \cite{MP2,HP} need to be modified. Note that since for $a\in\frac{1}{M}P$ and $\mu\in P $ the relation
$$e^{2\pi\i\sca{\mu}{a+Q^{\vee}}}= e^{2\pi\i\sca{\mu}{a}}$$
holds, the exponential mapping $e^{2\pi\i\sca{\mu}{x}}\in\Com$ for $\mu \in P$ and $x\in\frac{1}{M}P/Q^{\vee}$ is well-defined. 
\begin{tvr}
For all $\mu \in P$, $\mu\notin MQ^\vee$ there exists $x\in\frac{1}{M}P/Q^{\vee}$ such that $e^{2\pi\i\sca{\mu}{x}}\neq 1$.
\end{tvr}
\begin{proof}
Suppose there is some $\mu\in P$ and $\mu\notin MQ^\vee$ such that for all $\nu\in P$
\begin{equation}
 \sca{\mu}{\frac{\nu}{M}}=\sca{\frac{\mu}{M}}{\nu}\in \Z .
\end{equation}
Then from $\Z$-duality of $P$ and $Q^\vee$ follows that $\mu/M\in Q^\vee$, i.e. $\mu\in MQ^\vee$ -- a contradiction.
\end{proof}

\begin{cor}
For any $\mu \in P$ it holds that 
\begin{equation}\label{bdis}
 \sum_{y\in\frac{1}{M}P/Q^{\vee}} e^{2\pi i\sca{\mu}{y}}=\begin{cases} dM^n,\q \mu\in MQ^\vee  \\ 0,\q\q \mu\notin MQ^\vee.\end{cases}
\end{equation} 
\end{cor}

The scalar product of two functions $f,g:F^{\sigma}_{P,M}\map \Com$ is defined as
\begin{equation} \label{scp}
\sca{f}{g}_{F^{\sigma}_{P,M}}= \sum_{a\in F^{\sigma}_{P,M}}\ep(a) f(a)\overline{g(a)},
\end{equation}
where the numbers $\ep (a)$ are determined by (\ref{epR}). The following theorem shows that the sets of weights $\Lambda^{\sigma}_{P,M}$ are in one-to-one corrrespondence with the lowest maximal sets of pairwise orthogonal orbit functions.
\begin{thm}
For any $\la,\la' \in\Lambda^{\sigma}_{P,M}$ it holds that
\begin{equation}\label{ortho}
 \sca{\phi^\sigma_\la}{\phi^\sigma_{\la'}}_{F^{\sigma}_{P,M}}=d\abs{W}M^n h_M(\la) \delta_{\la,\la'},
\end{equation}
where $d$, $h_M(\la)$ were defined by (\ref{dnum}), (\ref{hM}), %(\ref{Sl},\ref{Ss})
respectively, and $|W|$ is the number of elements of the Weyl group~$W$.
\end{thm}
\begin{proof}Taking into account the set equality
$$F^{\sigma}_{P,M} \cup  \left[\frac{1}{M}P \cap H^\sigma   \right]= \frac{1}{M}P \cap F,   $$
the vanishing property \eqref{Fss} of $\phi^{\sigma}_\lambda$ on borders $H^\sigma$ gives the equality
\begin{equation*}
\sca{\phi^\sigma_\la}{\phi^\sigma_{\la'}}_{F^{\sigma}_{P,M}}= \sum_{a\in F^{\sigma}_{P,M}}\ep(a)\phi^\sigma_\la(a)\overline{\phi^\sigma_{\la'}(a)}=\sum_{a\in \frac{1}{M}P \cap F}\ep(a)\phi^\sigma_\la(a)\overline{\phi^\sigma_{\la'}(a)}.
\end{equation*}
The $W^{\mathrm{aff}}$-invariance properties \eqref{epshift} and \eqref{Sssym} imply the $W^{\mathrm{aff}}$-invariance of $\ep(a)\phi^\sigma_\la(a)\overline{\phi^\sigma_{\la'}(a)}$,
\begin{equation}\label{Waffin}
\ep(a)\phi^\sigma_\la(a)\overline{\phi^\sigma_{\la'}(a)}=\ep(w^{\mathrm{aff}}a)
\phi^\sigma_\la(w^{\mathrm{aff}}a)\overline{\phi^\sigma_{\la'}(w^{\mathrm{aff}}a)},\q w^{\mathrm{aff}}\in W^{\mathrm{aff}}.
\end{equation}
Using first the shift invariance in \eqref{Waffin}  with respect to shifts from $Q^\vee$ and \eqref{ept}, one gets  
$$\sum_{a\in \frac{1}{M}P \cap F}\ep(a)\phi^\sigma_\la(a)\overline{\phi^\sigma_{\la'}(a)}=\sum_{x\in \left[ \frac{1}{M}P/Q^\vee\right] \cap F}\wt\ep(x)\phi^\sigma_\la(x)\overline{\phi^\sigma_{\la'}(x)} $$
Secondly, the $W$-invariance in \eqref{Waffin} and relations \eqref{rfun1}, \eqref{rfun2} give
 $$ \sum_{x\in \left[ \frac{1}{M}P/Q^\vee\right] \cap F}\wt\ep(x)\phi^\sigma_\la(x)\overline{\phi^\sigma_{\la'}(x)} =\sum_{y\in\frac{1}{M}P/Q^\vee }\phi^\sigma_\la(y)\overline{\phi^\sigma_{\la'}(y)}.  $$
The $W$-invariance of $\frac{1}{M}P/Q^{\vee}$ allows to continue the calculation,
\begin{align}
\sum_{y\in\frac{1}{M}P/Q^\vee }\phi^\sigma_\la(y)\overline{\phi^\sigma_{\la'}(y)} = &\sum_{w'\in W}\sum_{w\in W} \sum_{y\in \frac{1}{M}P/Q^{\vee}}\sigma(ww')e^{2\pi\i\sca{w\la-w'\la'}{y}}\nonumber \\ = &\abs{W}\sum_{w'\in W}\sigma(w')\sum_{y \in \frac{1}{M}P/Q^{\vee}}e^{2\pi\i\sca{\la-w'\la'}{y}}. \label{ddd}
\end{align}
Note that $\la-w'\la'\in MQ^\vee$ for some $w'\in W$ means that $\la /M$ and $\la' /M$ lie in the same $W^{\mathrm{aff}}$-orbit. Since from definition \eqref{Ls} both $\la /M$ and $\la' /M$ are in $F$ and as the fundamental domain $F$ contains only one element from each $W^{\mathrm{aff}}$-orbit, $\la-w'\la'\in MQ^\vee$ implies $\la = \la'$. Thus if $\la \neq \la'$ then for all $w'\in W$ it holds that $\la-w'\la'\notin MQ^\vee$ and \eqref{bdis} forces $\sca{\phi^\sigma_\la}{\phi^\sigma_{\la'}}_{F^{\sigma}_{P,M}}=0$.

If $\la=\la'$ then in \eqref{ddd} are due to \eqref{bdis} non-zero summands only if $\la-w'\la\in MQ^\vee$, or equivalently $w'\in \psi \left(\mathrm{Stab}_{W^{\mathrm{aff}}}\left(\la/M\right)\right)$ and thus
\begin{align*}
\abs{W}\sum_{w'\in W}\sigma(w')\sum_{y \in \frac{1}{M}P/Q^{\vee}}e^{2\pi\i\sca{\la-w'\la}{y}} = & d\, |W|\, M^n\sum_{w'\in \psi\left(\mathrm{Stab}_{W^{\mathrm{aff}}}\left(\frac{\la}{M}\right)\right)}\sigma(w') . 
\end{align*}
Since for $w^{\mathrm{aff}}\in \mathrm{Stab}_{W^{\mathrm{aff}}}(\la/M)$ the property $\psi(w^{\mathrm{aff}})=1$  implies $w^{\mathrm{aff}}=1$ it holds that the subgroups
$\mathrm{Stab}_{W^{\mathrm{aff}}}\left(\la/M\right)$ and $\psi \left(\mathrm{Stab}_{W^{\mathrm{aff}}}\left(\la/M\right)\right) $ are isomorphic and thus it follows that
$$\sum_{w'\in \psi\left(\mathrm{Stab}_{W^{\mathrm{aff}}}\left(\frac{\la}{M}\right)\right)}\sigma(w')=\sum_{w^{\mathrm{aff}}\in \mathrm{Stab}_{W^{\mathrm{aff}}}\left(\frac{\la}{M}\right)}\sigma\circ \psi(w^{\mathrm{aff}}).$$
It remains to recall that $\la \in \Lambda^{\sigma}_{P,M}$ also means that $\la/M\in F^\sigma $ and together with definitions \eqref{Fsigma}, \eqref{hM} it gives 
$$\sum_{w^{\mathrm{aff}}\in \mathrm{Stab}_{W^{\mathrm{aff}}}\left(\frac{\la}{M}\right)}\sigma\circ \psi(w^{\mathrm{aff}})=\sum_{w^{\mathrm{aff}}\in \mathrm{Stab}_{W^{\mathrm{aff}}}\left(\frac{\la}{M}\right)}1= h_M(\la). $$   
\end{proof}

\vskip0.5cm\section{Discrete transforms, $S$-matrices and Kac-Peterson formula}

\subsection{Discrete orbit function transforms}\

The interpolating functions $I[f]^\sigma_M: \R^n\map\Com$ of any function $f:F^{\sigma}_{P,M}\map \Com$ are finite linear combinations 
\begin{align}
I[f]^\sigma_M(a):= \sum_{\la\in \Lambda^{\sigma}_{P,M}} c^\sigma_\la \phi^\sigma_\la(a) \label{intc}
\end{align}
such that
\begin{align}\label{intcs}
I[f]^\sigma_M(a)=& f(a), \q a\in F^{\sigma}_{P,M}.
\end{align}
The discrete orthogonality \eqref{ortho} and the completeness \eqref{complete} of the sets of functions $\phi^\sigma_\la(a)$, $\la \in \Lambda^{\sigma}_{P,M}$, $a \in F^{\sigma}_{P,M}$ ensure that  the coefficients $c^\sigma_\la$ are uniquely determined. The formulas for calculation of $c^\sigma_\la$ which constitute discrete orbit function transforms, are given by 
\begin{align}\label{trans}
c^\sigma_\la=& \frac{\sca{f}{\phi^\sigma_\la}_{F^{\sigma}_{P,M}}}{\sca{\phi^\sigma_\la}{\phi^\sigma_\la}_{F^{\sigma}_{P,M}}}=(d\abs{W} M^n h_M(\la))^{-1}\sum_{a\in F^{\sigma}_{P,M}}\ep(a) f(a)\overline{\phi^\sigma_\la(a)}
\end{align}
and the corresponding Plancherel formulas are of the form
\begin{align*}
\sum_{a\in F^{\sigma}_{P,M}} \ep(a)\abs{f(a)}^2 = & d \abs{W} M^n \sum_{\la\in\Lambda^{\sigma}_{P,M}}h_M(\la)|c^\sigma_\la|^2 .
\end{align*}

\subsection{$S$-matrices and Kac-Peterson formula}\  

For every simple Lie algebra $X_n$, there exists an untwisted affine Kac-Moody algebra $X^{(1)}_n$ with a horizontal subalgebra isomorphic to $X_n$ \cite{Kac}.  When the central element of $X_n^{(1)}$ is fixed to be a positive integer, the so-called level $k\in\N$, the integrable representations of the affine algebra are in one-to-one correspondence with the horizonal weights of $F_{P,M}$, with $M=k+g$. 
Kac and Peterson \cite{KP} discovered that the characters of these representations of $X^{(1)}_n$, the affine characters, form a finite-dimensional representation of the modular group $SL(2; \Z)$. 

Rational conformal field theory \cite{DMS} is the physical context for this result. Like all rational conformal field theories, the WZNW models can be formulated on any Riemann surface with a finite number of marked points. The corresponding correlation functions are written in terms of so-called conformal blocks, labelled by a trivalent graph that is a  degeneration of the marked Riemann surface. Since the trivalent graph is not unique, a linear transformation between the different sets of  conformal blocks must exist, so that the (physical) correlation function is unique \cite{MS}. 

If a  WZNW conformal field theory is considered on a torus, the resulting correlation function is the partition function, a sesquilinear combination of affine characters. These affine characters are the conformal blocks of the torus, labelled by a trivial trivalent graph, the circle.  But a circle can be obtained as the degenerate limit of the torus in an infinite number of different ways -- as the $a$ and $b$ cycles of the torus, e.g. The different choices are mapped to each other by the action of the modular group, with generators conventionally denoted $S$ and $T$. Therefore, there is a representation of the modular group of dimension equal to $|F_{P,k+g}|$.  In particular, the generator $S$ is represented as a unitary, symmetric matrix of the same dimension.  

Theorem 6.3 above proves the orthogonality (\ref{ortho}) of the Weyl-orbit functions $\varphi_\lambda^\sigma$.  The following unitary matrices are then easily constructed:
\begin{align}\label{sigmaKP} 
S^\sigma_{\la,\mu}=\frac{i^\frac{|\Pi|}{2}\phi^\sigma_\la (\frac{-\mu}{k+q^\sigma})}{\sqrt{d (k+q^\sigma)^n h_{k+q^\sigma}(\la)h_{k+q^\sigma}(\mu)}},\q \la,\mu \in \Lambda^{\sigma}_{P,k+q^\sigma}\ . 
\end{align}
Like the Kac-Peterson affine modular-$S$ matrices \cite{KP}, the matrices with entries $S^\sigma_{\la,\mu}$ are unitary and symmetric.  Indeed,  a simple check reveals that $S^{\sigma^e}_{\la,\mu}$ are precisely the elements of the standard Kac-Peterson matrices.  With the new weight-lattice discretization, the identity of the affine modular $S$-matrix and Weyl-orbit functions is established. 

When $\sigma\not=\sigma^e$, the matrices with elements (\ref{sigmaKP}) can be considered a generalization of the Kac-Peterson $S$ matrix.  We do not know if they are relevant to WZNW conformal field theories. We mention here, however, that discretized versions of the Weyl-orbit functions for both $\sigma=\sigma^e$ and $\sigma=\id$ were used in \cite{GJW}.

%%%%%%%%%%%%%%%%%%%%%%%%%%%%%%%%%%%%%%%
\vskip0.5cm\section{Conclusion} 

Let us first point out where to find the main results of this paper.  Formulas (\ref{FP}-\ref{complete}) establish the argument-label symmetry of the orbit functions in the new weight-lattice discretization.  The orthogonality of the orbit functions is established in Theorem 6.3. The affine $S^\sigma$ matrices defined by (\ref{sigmaKP}) generalize the Kac-Peterson $S=S^{\sigma^e}$ matrix, and establish the identity of the latter with weight-discretized Weyl-orbit functions. 

A comparison is in order of the new weight-lattice discretization with the coweight-lattice discretization of \cite{HP}.  The differences are made plain by comparing Figures 1 and 2 here with Figures 2 and 3 in \cite{HP}, respectively. All four diagrams treat the case of algebra $C_2$, and the difference in resolutions used ($M=3$ here and $M=4$ in Figs. 3 and 4 of \cite{HP}) does not obscure.  Another difference is that while Figs. 1 and 2 here deal with fragments $F^\sigma_{P,M}$ for all sign homomorphisms $\sigma$, ref. \cite{HP} only treats the $P^\vee$-fragment $F_M$. The fragments $F_M$ in \cite{HP}, and their generalizations $F^\sigma_M$, should now be denoted $F_{P^\vee, M}$, and $F^\sigma_{P^\vee, M}$, respectively. But a comparison of the two discretizations in the $\sigma=\id$ case will make the essential distinctions clear.  

There is a difference between $P^\vee$- and $P$-discretizations in the $C_2$ case because there is a short simple root $\alpha_1$ and so a corresponding long fundamental coweight $\omega_1^\vee = 2\omega_1$.  $F_{P, M}$ is a finer discretization of the fundamental region of $W^{\mathrm{aff}}$ than is $F_{P^\vee, M}$.  From Fig.~2 of \cite{HP}, we see that while the $P^\vee$-discretization chops the boundary from 0 to $\omega^\vee_2/m_2$ into $M$ segments, only $M/m_1=M/2$ segments occur between 0 and $\omega^\vee_1/m_1$, the vertex of $F$. On the other hand, the same diagram shows that there are $M=4$ segments between  0 and $\omega_1^\vee$, a vertex of $F^\vee$.  $P^\vee$-discretization creates a more `uniform' fragment of $F^\vee$ than of $F$. Consequently, as Fig.~3 illustrates, the orbit-function labels discretize a dilation $MF^\vee=4F^\vee$ of $F^\vee$, rather than of $F$. 

In contrast, Fig.~1 here shows that the $P$-discretization yields the same number of segments, $M=3$, on the edges between any 2 of the 3 vertices of $F$.  The result is illustrated in Fig.~2: the resulting Weyl-orbit function labels lie on a $P$-fragment of the $M$-dilation ($M=3$) of $F$. By factoring out the dilation size, $M$, both arguments and labels in the $P$-discretization can be treated on the same footing.  
Table~2 compares the treatment of arguments and labels in coweight- and weight-discretization of Weyl-orbit functions. 

{
\renewcommand{\arraystretch}{1.25}
\begin{table}
\begin{tabular}{|c||c|c|}
\hline
Discretization       & Arguments\quad & Labels\quad  \\
\hline\hline
Coweight & $P^\vee/M\,\cap\, F^\sigma$  & $P\,\cap\, MF^{\sigma\vee}$ \\ \hline
Weight & $P/M\,\cap\, F^\sigma$  & $P\,\cap\, MF^{\sigma}$ \\ \hline
\end{tabular}
\medskip
\caption{Weight- vs. coweight discretization of Weyl-orbit functions.}\label{tab2}
\end{table}
}

The results presented here describe another practicable discretization of Weyl-orbit functions of all types.  Besides coroot- and root-lattice discretizations, it is also likely the only other type of discretization that is natural from the Lie-theoretic point of view.  In addition, the new discretization identifies orbit $S$-functions and the affine Kac-Peterson modular $S$-matrices. The discovery of new properties of discretized orbit functions analogous to those of affine modular matrices (see \cite{Ga, Ga02, GW, GRW}, e.g.) that was begun in \cite{HW} should now continue easily. We also hope that the study of Weyl-orbit functions will impact conformal field theory.

%\begin{itemize}
%\item
%\end{itemize}

%%%%%%%%%%%%%%%%%%%%%%%%%%%%%%%%%%%%%%%%%%%%%
\vskip0.5cm\section*{Acknowledgments}
Research support was provided by RVO68407700 (JH), a Discovery Grant from NSERC (MAW), and by the University of Lethbridge.

\end{document}